\documentclass[runningheads]{llncs}

\usepackage{amsmath}
\usepackage{amssymb}

\usepackage{latexsym}
\usepackage{enumerate}
\usepackage[usenames,dvipsnames]{xcolor}
\usepackage{cite}

\usepackage{graphicx}
\usepackage{tikz}
\usetikzlibrary{calc,positioning,matrix,fit}

\usepackage[compatibility=false]{caption}
\usepackage[labelformat=parens]{subcaption}

\usepackage{color}

\usepackage{hyperref}
\hypersetup{
  colorlinks=true,       %
  linkcolor=blue!50!black
}

\newcommand{\IC}{\textsf{IC}}
\newcommand{\NIC}{\textsf{NIC}}

\newcommand{\One}{\textsf{1}}

\newcommand{\ie}{i.\,e.}
\newcommand{\eg}{e.\,g.}

\newcommand{\Edge}[2]{\{#1,#2\}}
\newcommand{\Emb}{\mathcal{E}}
\newcommand{\Embedding}[1]{\Emb{}(#1)}

\newcommand{\TabLabel}[1]{\label{tab:#1}}

\newcommand{\Fig}[1]{Fig.\,\ref{fig:#1}}

\newcommand{\Tab}[1]{Table~\ref{tab:#1}}

\newcommand{\Lem}[1]{Lemma~\ref{lem:#1}}

\newcommand{\Sect}[1]{Sect.\,\ref{sect:#1}}

\begin{document}

\title{%
A Note on \IC{}-Planar Graphs %
}

\author{%
Christian Bachmaier \and
Franz J.\,Brandenburg  \and
Kathrin Hanauer}

\institute{%
University of Passau,
94030 Passau, Germany \\
\email{\{bachmaier|brandenb|hanauer\}@fim.uni-passau.de}
}

\authorrunning{C.\ Bachmaier, F.\ J.\ Brandenburg, K.\ Hanauer}

\maketitle

\begin{abstract}
  A graph is \IC{}-planar if it admits a drawing in the plane with at most
  one crossing per edge and such that two pairs of crossing edges share
  no common end vertex. \IC{}-planarity specializes both \NIC{}-planarity,
  which allows a pair of crossing edges to share at most one vertex,
	and \One{}-planarity, where each edge may be crossed at most once.

  We show that there are infinitely maximal \IC{}-planar graphs with
  $n$ vertices and $3n-5$ edges and thereby prove a tight lower bound on
  the density of this class of graphs.
\end{abstract}

\section{Introduction}\label{sect:introduction}

A graph $G$ is \emph{maximal} in a graph class $\mathcal{G}$ if no edge can be
added to $G$ without violating the defining class. The
\emph{density} (\emph{sparsity}) of $\mathcal{G}$ is an upper (lower) bound
on the number of edges of maximal graphs $G \in \mathcal{G}$ with $n$
vertices.
A maximal graph $G$ is \emph{densest} (\emph{sparsest}) in $ \mathcal{G}$
if its number of edges meets the upper (lower) bound.

It is well-known that every maximal planar graph is triangulated and has $3n-6$
edges.
The densest and sparsest planar graphs coincide. This does not longer hold
for \One{}-planar graphs, which are graphs that can be drawn with
at most one crossing per edge. These graphs have recently received much
interest~\cite{klm-ab1p-17}. \One{}-planar graphs with $n$
vertices have at most $4n - 8$ edges, and this bound is tight for $n = 8$
and all $n \geq 10$~\cite{bsw-bs-83, b-4mg1p-15}. However, there are
sparse maximal \One{}-planar graphs with less than $2.65 n$
edges~\cite{cm-aoepl-13}. The best known lower bound on the sparsity of
\One{}-planar graphs is $2.22 n$~\cite{bsw-bs-83} and neither the upper nor
the lower bound are known to be tight.

There are some subclasses of \One{}-planar graphs with different bounds for
the density and sparsity. A graph is \emph{\IC{}-planar} (independent
crossing planar)~\cite{a-cbirc-08, ks-cpgIC-10, zl-spgic-13, bdeklm-IC-16}
if it admits a drawing with at most one crossing per edge so that each
vertex is incident to at most one crossing edge, and \emph{\NIC{}-planar}
(near independent crossing planar) if two pairs of crossing edges share at
most one vertex~\cite{zl-spgic-13}. \IC{}-planar graphs have an upper bound
of $3.25 n - 6$ on the number of edges, which is known as a tight bound,
since there are such graphs for all $n = 4k$ and $k \geq 2$~\cite{zl-spgic-13}.
The lower bound has not been addressed yet.
The upper and lower bounds on the density of \NIC{}-planar graphs are
$\frac{18}{5}(n - 2)$~\cite{zl-spgic-13} and $\frac{16}{5}(n -
2)$~\cite{bbhnr-npg-17, cs-tc1pg-14} and both bounds are known to be tight for
infinitely many values of $n$. Outer \One{}-planar graphs are another
subclass of \One{}-planar graphs. They must admit a \One{}-planar embedding
such that all vertices are in the outer face~\cite{abbghnr-o1p-15,heklss-ltao1p-15}.
Results on the density of maximal graphs are summarized
in \Tab{density}.

\begin{table}[tb]
\caption{The density of maximal graphs on $n$ vertices.}
\TabLabel{density}
\setlength{\tabcolsep}{3pt}
\renewcommand{\arraystretch}{1.2}
\centering
\begin{tabular}{r|c|c|c|c}
\noalign{\smallskip}
& 1-planar
& \NIC{}-planar
& \IC{}-planar
& outer 1-planar
\\ \hline
upper bound
& $4n - 8$~\cite{bsw-bs-83,bsw-1og-84}
& $\frac{18}{5}(n - 2)$~\cite{z-dcmgprc-14, bbhnr-npg-17}
& $\frac{13}{4}n - 6$~\cite{ks-cpgIC-10}
& $\frac{5}{2}n - 2$~\cite{abbghnr-o1p-15}
\\
$\lfloor$ example & $4n-8$~\cite{bsw-bs-83, bsw-1og-84} &
$\frac{18}{5}(n - 2)$~\cite{cs-tc1pg-14, bbhnr-npg-17} &
$\frac{13}{4}n - 6$~\cite{zl-spgic-13} & $\frac{5}{2}n -
2$~\cite{abbghnr-o1p-15}
\\
lower bound
& $\frac{20}{9} n - \frac{10}{3}$~\cite{bt-idm1p-15}
& $\frac{16}{5}(n-2)$~\cite{bbhnr-npg-17}
& $3n-5$~[\Sect{density}]
& $\frac{11}{5}n - \frac{18}{5}$~\cite{abbghnr-o1p-15}
\\
$\lfloor$ example
& $\frac{45}{17}n - \frac{84}{17}$~\cite{begghr-odm1p-13}
& $\frac{16}{5}(n - 2)$~\cite{bbhnr-npg-17}
& $3n-5$~[\Sect{density}]
& $\frac{11}{5}n - \frac{18}{5}$~\cite{abbghnr-o1p-15}
\\
\end{tabular}
\end{table}
Here, we establish a lower bound of $3n - 5$ on the density of \IC{}-planar
graphs and show that it is tight for all $n \geq 5$.

\section{Density}\label{sect:density}
We first prove the existence of maximal \IC{}-planar graphs that have
$n$ vertices and only $3n -5$ edges:
\begin{lemma}\label{lem:ic-lowbound}
For every $n \geq 5$ there is a maximal \IC{}-planar graph with $n$ vertices
and $3n-5$ edges.
\end{lemma}
\begin{proof}
As $K_5$ has exactly $3n-5 = 10$ edges and is \IC{}-planar, the statement
trivially follows for $n = 5$.
Let us hence assume in the following that $n \geq 6$.
We construct a graph $G_n$ with $n$ vertices and $3n -5$ edges as follows:
$G_n$ consists of $n-2$ vertices forming a circle $C = (v_0, v_1, \dots,
v_{n-3})$ as well as two pole vertices $p$ and $q$.
For every $0 \leq i < n-2$, $G_n$ has edges $\Edge{v_i}{p}$ and
$\Edge{v_j}{q}$.
Additionally, there is an edge $\Edge{p}{q}$ connecting the poles.
As an example, \Fig{sparse-ic} depicts the graph $G_8$.
Then, every vertex $v_i$, $0 \leq i < n-2$, is incident to exactly two circle
edges as well as to both $p$ and $q$, and $p$ and $q$ are each incident to
$n-1$ edges.
Hence, $G_n$ has $\frac{1}{2}(4(n-2) + 2(n-1)) = 3n-5$ edges.
As every planar graph has at most $3n-6$ edges, every embedding must contain at
least one pair of crossing edges.
Let $\Embedding{G_n}$ be any \IC{}-planar embedding of $G_n$.
We will now show that $\Edge{p}{q}$ must be crossed in $\Embedding{G_n}$
and that $\Embedding{G_n}$ is unique up to isomorphism.
\begin{figure}[tb]
\centering
\begin{tikzpicture}
\node (graph) {\includegraphics[scale=.75]{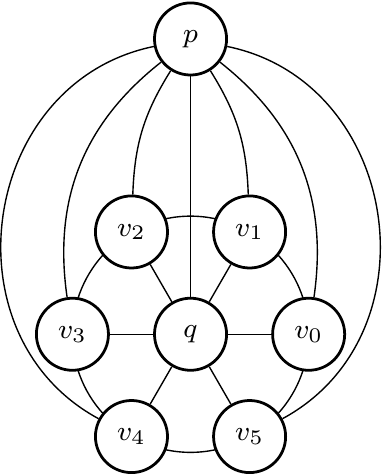}
  \phantomsubcaption\label{fig:sparse-ic}};
\node (nopq) [right=of graph] {\includegraphics[scale=.75]{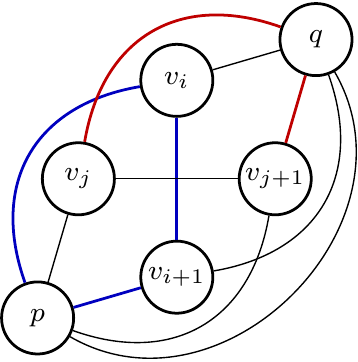}
  \phantomsubcaption\label{fig:sparse-ic-nopq}};
\node (justp) [below=of graph] {\includegraphics[scale=.75]{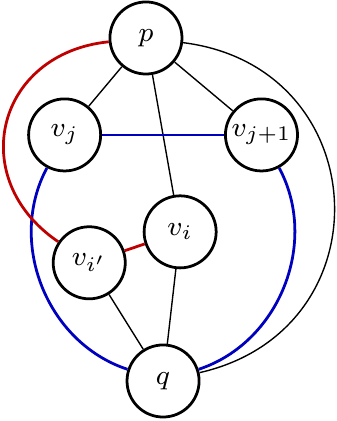}
  \phantomsubcaption\label{fig:sparse-ic-justp}};
\node (bothpq) [right=of justp] {\includegraphics[scale=.75]{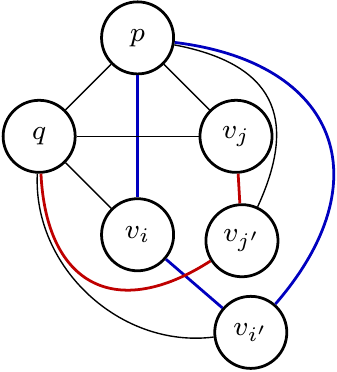}
  \phantomsubcaption\label{fig:sparse-ic-bothpq}};
\node[anchor=north west] at (graph.north west) {(\subref*{fig:sparse-ic})};
\node[anchor=north west] at (nopq.north west) {(\subref*{fig:sparse-ic-nopq})};
\node[anchor=north west] at (justp.north west) {(\subref*{fig:sparse-ic-justp})};
\node[anchor=north west] at (bothpq.north west) {(\subref*{fig:sparse-ic-bothpq})};
\end{tikzpicture}
\caption{%
Proof of \Lem{ic-lowbound}:
The graph $G_n$ for $n = 8$~(\subref{fig:sparse-ic}) along with sketches of the
cases where
an edge $\Edge{v_i}{v_{i+1}}$ crosses an edge $\Edge{v_j}{v_{j+1}}$~(\subref{fig:sparse-ic-nopq}),
$\Edge{v_i}{p}$ crosses an edge $\Edge{v_j}{v_{j+1}}$~(\subref{fig:sparse-ic-justp}),
and
$\Edge{v_i}{p}$ crosses an edge $\Edge{v_j}{q}$~(\subref{fig:sparse-ic-bothpq}),
which all yield non-\IC{}-planar embeddings of $G_n$.
}%
\label{fig:sparsest-ic}
\end{figure}

Suppose that an edge $\Edge{v_i}{v_{i+1}}$ crosses an edge
$\Edge{v_j}{v_{j+1}}$ in $\Embedding{G_n}$ (see \Fig{sparse-ic-nopq}).
Due to \IC{}-planarity,
$\Edge{v_i}{p}$, $\Edge{v_{i+1}}{p}$,
$\Edge{v_j}{q}$, and $\Edge{v_{j+1}}{q}$
must be planar.
In consequence of the crossing, $v_j$ and $v_{j+1}$ lie on different sides of
the closed path $P$ consisting of $\Edge{v_i}{p}$, $\Edge{v_i}{v_{i+1}}$, and
$\Edge{v_{i+1}}{p}$.
Hence, either $\Edge{v_j}{q}$ or $\Edge{v_{j+1}}{q}$ must cross an edge of $P$,
a contradiction.
Thus, every crossing in $\Embedding{G_n}$ must involve at least one of $p$ or $q$.

Suppose that an edge $\Edge{v_i}{p}$ crosses an edge $\Edge{v_j}{v_{j+1}}$ in
$\Embedding{G_n}$ (see \Fig{sparse-ic-justp}).
By \IC{}-planarity, $\Edge{v_j}{q}$ and $\Edge{v_{j+1}}{q}$ must be planar.
As $n \geq 6$, there must be a vertex $v_{i'}$ adjacent to $v_i$.
Furthermore, $\Edge{v_{i'}}{v_i}$ and $\Edge{v_{i'}}{p}$ must be planar
due to \IC{}-planarity.
In consequence of the crossing, $p$ and $v_i$ however lie on different sides
of the closed path $P$ consisting of $\Edge{v_j}{q}$, $\Edge{v_j}{v_{j+1}}$,
and $\Edge{v_{j+1}}{q}$, so either $\Edge{v_{i'}}{v_i}$ or $\Edge{v_{i'}}{p}$
must cross an edge of $P$, a contradiction.
Thus, every crossing in $\Embedding{G_n}$ must involve both $p$ and $q$.

Finally, suppose that an edge $\Edge{v_i}{p}$ crosses an edge $\Edge{v_j}{q}$
in $\Embedding{G_n}$ (see \Fig{sparse-ic-bothpq}).
As $n \geq 6$, there must be two further vertices $v_{i'} \neq v_{j'}$
such that $v_{i'}$ is adjacent to $v_i$ and $v_{j'}$ is adjacent to $v_j$.
Furthermore, $\Edge{v_i}{v_{i'}}$, $\Edge{v_j}{v_{j'}}$, $\Edge{v_{i'}}{p}$,
and $\Edge{v_{j'}}{q}$ must be planar due to \IC{}-planarity.
In result of the crossing, $v_j$ and $q$ lie on different sides of the
closed path $P$ consisting of $\Edge{v_i}{p}$, $\Edge{v_i}{v_{i'}}$,
and $\Edge{v_{i'}}{p}$.
Thus, one of $\Edge{v_j}{v_{j'}}$ and $\Edge{v_{j'}}{q}$ must cross $P$,
a contradiction.

Consequently, every crossing in $\Embedding{G_n}$ must contain $\Edge{p}{q}$,
which in turn can only cross an edge $\Edge{v_i}{v_{i+1}}$ for some $i$
with $0 \leq i < n-3$ or $\Edge{v_{n-3}}{v_0}$, as depicted, \eg, in \Fig{sparse-ic}.
As no edge can be added to $G_n$ and $\Embedding{G_n}$ such that
\IC{}-planarity is preserved, $G_n$ is maximal.
\qed{}
\end{proof}
Note that the embedding of the graphs $G_n$ from the proof of \Lem{ic-lowbound}
are unique up to isomorphism, because $\Edge{p}{q}$ must cross an arbitrary edge
$\Edge{v_i}{v_{i+1}}$.
Concerning the upper bound, observe that every maximal \IC{}-planar graph with
$n \geq 5$ vertices also has at least $3n - 5$ edges, as there always is at
least one pair of crossing edges, \ie, it cannot be planar.
\begin{theorem}
For all $n \geq 5$, maximal \IC{}-planar graphs with $n$ vertices have at
least $3n - 5$ edges, and this bound tight.
\end{theorem}

\bibliographystyle{splncs03}
\bibliography{paper}
\end{document}